\documentclass[11pt,a4paper]{article}
\usepackage{algcompatible,algpseudocode,a4wide}
\usepackage{amsmath,amsthm,amssymb,amsfonts,mathdots,authblk,algorithm}
\usepackage{graphicx,url}
\usepackage{tikz}
\title{The RAM equivalent of P vs.\ RP}
\author{Michael Brand\\
\texttt{michael.brand@alumni.weizmann.ac.il}}
\affil{Faculty of IT, Monash University\\
Clayton, VIC 3800\\
Australia}
\date{\today}

\makeatletter
\newcommand*{\clr}{%
  \nonscript\mskip-\medmuskip\mkern5mu%
  \mathbin{\mathsf{clr}}\penalty900\mkern5mu%
  \nonscript\mskip-\medmuskip
}
\makeatother

\makeatletter
\newcommand*{\modop}{%
  \nonscript\mskip-\medmuskip\mkern5mu%
  \mathbin{\mathsf{mod}}\penalty900\mkern5mu%
  \nonscript\mskip-\medmuskip
}
\makeatother

\newcommand{\set}{\mathop{\textit{SET}}}
\newcommand{\BigO}{\mathop{\text{O}}}

\newcommand{\Th}{\mathop{\Theta}}

\newcommand{\defeq}{\stackrel{\text{def}}{=}}
\newcommand{\minus}{\mathop{\tikz{
\coordinate (P1) at (0,0.1);
\coordinate (P2) at (0.2,0.1);
\coordinate (P3) at (0.1,0.18);
\coordinate (P4) at (0,0);
\coordinate (P5) at (0.2,0.2);
\draw[fill=white, color=white] (P4) rectangle (P5);
\draw[-] (P1) -- (P2);
\fill (P3) circle(1pt);
}}}

\newtheorem{thm}{Theorem}

\newtheorem{lem}{Lemma}
\newtheorem{cor}{Corollary}[thm]

\theoremstyle{definition}
\newtheorem{defi}{Definition}

\begin{document}
\maketitle

\begin{abstract}
One of the fundamental open questions in computational complexity is whether
the class of problems solvable by use of stochasticity under the Random
Polynomial time (RP) model is larger than the class of those solvable in
deterministic polynomial time (P). However, this question is only open for
Turing Machines, not for Random Access Machines (RAMs).

Simon (1981) was able to show that for a sufficiently equipped Random Access
Machine, the ability to switch states nondeterministically does not entail any
computational advantage. However, in the same paper, Simon describes a
different (and arguably more natural) scenario for stochasticity under the RAM
model. According to Simon's proposal, instead
of receiving a new random bit at each execution step, the RAM program is able to
execute the pseudofunction $\textit{RAND}(y)$, which returns a uniformly
distributed random integer in the range $[0,y)$. Whether the ability to
allot a
random integer in this fashion is more powerful than the ability to allot a
random bit remained an open question for the last $30$ years.

In this paper, we close Simon's open problem, by fully characterising the class
of languages recognisable in polynomial time by each of the RAMs regarding
which the question was posed. We show that for some of these, stochasticity
entails no advantage, but, more interestingly, we show that for others it
does.
\end{abstract}

\section{Introduction}
The Turing machine (TM), first introduced in \cite{Turing:Computable}, is
undoubtedly the most familiar computational model. However, for algorithm
analysis it often fails to adequately represent real-life complexities, for
which reason the random access machine (RAM),
closely resembling the intuitive
notion of an idealised computer, has become the common choice in
algorithm design. Ben-Amram and Galil \cite{Galil:Shift} write
``The RAM is intended to model what
we are used to in conventional programming, idealized in order to be better
accessible for theoretical study.''

Here, ``what we are used to in conventional programming'' refers, among other
things, to the ability to manipulate high-level objects
by basic commands. However, this ability comes with some
unexpected side effects. For example, it was shown regarding many RAMs
working with fairly limited instruction sets that they are able to
recognise any PSPACE problem in deterministic polynomial time
\cite{Schonhage:rams, Simon:Multiplication, Bertoni:pTime_RAM, Pratt:VMs}.
A unit-cost RAM equipped only with arithmetic operations, Boolean operations and
bit shifts can, in fact, recognise in constant time any language that is
recognised by a TM in time and/or space constrained by \emph{any} elementary
function of the input size \cite{Brand:ALNs}. In polynomial time, such
a RAM can recognise a class of languages which we denote PEL and define below.

However, just like $\text{PSPACE}=\text{NPSPACE}$, and for basically the same
reasons, it was shown that nondeterminism does not make any of these RAM models
more powerful \cite{Pratt:VMs, Simon:feasible, Simon:division2}.

The class
of problems that can be solved using \emph{stochastic} computation in the RP
model
is an intermediate class between deterministic computation and nondeterministic
computation. As such, the results that nondeterminism confers no advantage
apply also to stochasticity.

The basic idea is that any machine that spans all of PSPACE can enumerate over
an exponential number of possibilities (denoted by a polynomial number of bits),
and can therefore simulate any assignment of a polynomial number of stochastic
or nondeterministic bits.

Simon \cite{Simon:division2}, however, suggested a different approach to the
definition of stochasticity. Because RAMs work natively with nonnegative
integers, and because they derive their power from this ability to manipulate
in constant time large operands, it seems unnatural that stochasticity in
them will be limited to bits alone. Instead, Simon suggested to equip the RAM
with a pseudofunction, $\textit{RAND}(y)$, whose output is an integer random
variable, $x$, uniformly distributed in the range $0\le x<y$.

The question of whether polynomial time computation using this pseudofunction,
and using the acceptance criteria of RP (no input is falsely accepted,
and any input that is in the language has at most probability $1/2$ of being
falsely rejected), a model that we refer to as RP-RAM, is a more powerful
model than the deterministic PTIME-RAM, remained open for the past $30$ years,
and is the RAM equivalent of the Turing machine question
``$P\stackrel{?}{=}RP$''.

This paper analyses specific examples of RAMs whose native operation set is a
subset of the arithmetic, bitwise Boolean and bit-shift operations, and shows
that for
some of these stochasticity is advantageous and for others not. The RAMs
considered are those analysed by Simon \cite{Simon:division2} as well as by
Simon and Szegedy \cite{Simon:RAM_w_various}. To the best of
the author's knowledge, the examples given here where stochasticity confers
an advantage are the first known examples for natural, general-purpose
computational models in which
stochastic computation under RP acceptance criteria is provably stronger than
deterministic polynomial time computation.
(The results most comparable to it, in this respect, are those of
Heller \cite{Heller:relativized}, in the context of comparing relativised
polynomial hierarchies with relativised RP hierarchies.)

In order to present our results in greater detail, we first redefine, briefly,
the RAM model.
(See \cite{Aho:Algorithms} for a more formal introduction.)

A Random Access Machine, denoted $\text{RAM}[\textit{op}]$,
is a computation model that affords all that we expect from a modern computer
in terms of flow control (loops, conditional jump instructions, etc.) and
access to variables (direct and indirect addressing). The operations it can
perform are those belonging to the set $\textit{op}$, and these are assumed
to execute in a single unit of time each. A comparator for equality is also
assumed to be available, and this also executes in a single unit of
time. The variables (or \emph{registers})
of a RAM contain nonnegative integers and are also indexable by addresses that
are nonnegative integers. 

Because the result of RAM operations must be storable in a register, and
therefore be a nonnegative integer, operations such as subtraction (which may
lead to a negative result) cannot be supported directly. Instead, we use
\emph{natural subtraction}\footnote{Natural subtraction will be taken here to
share the same
properties as normal subtraction in terms of operator precedence, etc..},
denoted ``$\minus$'' and defined as
\[
a\minus b\defeq \max(a-b,0).
\]

By the same token, regular bitwise negation is not allowed, and
$\lnot a$ is tweaked to mean that the bits of $a$ are negated only up to and
including its most significant ``$1$'' bit\footnote{We define $\lnot 0$ to be
zero.}. We use $X \clr Y$ to denote
what would have been $X \land \lnot Y$ if the standard ``$\lnot$'' operation
had been available.

Furthermore, we assume, following e.g.\ \cite{Mansour:floor1}, that all explicit
constants used as operands in RAM programs belong to the set $\{0,1\}$.
This assumption does not make a material difference to the results, but
it simplifies the presentation.

In order for our RAMs to be directly comparable to TMs, we consider only RAMs
that take a single integer as input. This is the value of some pre-designated
register (e.g.\ $R[0]$) at program start-up. All other registers are
initialised to zero.
A RAM is considered to ``accept'' the input if the RAM program ultimately
terminates, and the final value of a pre-designated output register
(e.g.\ $R[0]$) is
nonzero.

The flavour of Turing machine most directly comparable to this design is one
that works on a one-sided-infinite tape over a binary
alphabet, where ``$0$'' doubles as the blank. All TMs used in this paper are
of this type.

To introduce our results, we first define the function class PEL and the
eponymous complexity class.

\begin{defi}[PEL]
PEL (standing for ``Polynomial Expansion Limit'')
is the class of functions that can be described by
\[
f(n)={}^{p(n)}2,
\]
where $p$ can be any polynomial and the
left superscript denotes tetration.

The complexity class PEL is the class of languages recognisable by a TM
in time $f(n)$, where $f(n)$ is in PEL and $n$ is the bit-length of the
input.  Equivalently, PEL is the class of languages recognisable by a TM
working on a tape of size $f(n)$, where $f(n)$ is in PEL.

PEL can therefore also be described as either PEL-TIME or PEL-SPACE.
\end{defi}

The equivalence of the time-constrained and the tape-constrained definitions
is given by the well-known relation
$\text{$f$-TIME}\subseteq\text{$f$-SPACE}\subseteq\text{$\exp(f)$-TIME}$.

Our results are as follows.

\begin{thm}\label{T:RPRAM}
\[
\text{RP-RAM}[+,[\minus],[\times],[\div],\leftarrow,[\rightarrow],\textit{Bool\/}]
=\text{PEL}.
\]
\end{thm}

In the above theorem, operations that appear in brackets inside the operation
list are optional, in the sense that the theorem continues to be true both
with and without them. The operation ``$\div$'' is integer division,
``$\leftarrow$'' is left shifting ($a\leftarrow b \defeq a \times 2^b$),
``$\rightarrow$'' is right shifting
($a\rightarrow b \defeq \lfloor a/2^b \rfloor$) and ``$\textit{Bool\/}$'' is
shorthand for the set of all bitwise Boolean operations.\footnote{For
assignment, we use ``$\Leftarrow$'', so as to disambiguate it from
left shifting.}

Two direct corollaries from Theorem~\ref{T:RPRAM} are

\begin{cor}\label{C:wo_div}
\[
\text{RP-RAM}[+,[\minus],[\times],\leftarrow,[\rightarrow],\textit{Bool\/}]
\neq\text{P-RAM}[+,[\minus],[\times],\leftarrow,[\rightarrow],\textit{Bool\/}].
\]
\end{cor}

\begin{cor}\label{C:w_div}
\[
\text{RP-RAM}[+,[\minus],[\times],/,[\div],\leftarrow,[\rightarrow],\textit{Bool\/}]
=\text{P-RAM}[+,[\minus],[\times],/,[\div],\leftarrow,[\rightarrow],\textit{Bool\/}].
\]
\end{cor}

Here, ``$/$'', known as \emph{exact division}, is a weaker form of division.
The result of $a/b$ is the same as that of integer division (``$a\div b$''),
but $a/b$ is only defined when $a$ is a multiple of $b$.

These corollaries can be derived from Theorem~\ref{T:RPRAM} by making use of
the following known facts.

\begin{thm}[\cite{Brand:ALNs}]
\[
\text{P-RAM}[+,[\minus],[\times],/,[\div],\leftarrow,[\rightarrow],\textit{Bool\/}]=\text{PEL}
\]
\end{thm}
and
\begin{thm}[\cite{Brand:ALNs,Simon:RAM_w_various}]
\[
\text{P-RAM}[+,[\minus],[\times],\leftarrow,[\rightarrow],\textit{Bool\/}]=\text{PSPACE},
\]
\end{thm}
given that $\text{PEL}\neq\text{PSPACE}$ is known from \cite{Stearns:hierarchies, Geffert:space_hierarchy, Seiferas:space}.

Together, Corollaries~\ref{C:wo_div} and \ref{C:w_div} show the surprising fact
that for unit-cost RAMs the answer to the P vs.\ RP question can go either way,
depending on the choice of a basic operation set: by simply adding division as
a basic operation, the answer is reversed.

The rest of this paper is dedicated to a proof of Theorem~\ref{T:RPRAM}
and some corollaries.
It is arranged as follows.

We begin, in Section~\ref{S:preliminaries}, by constructing the
basic tools used in the proof.

The main body of the proof of Theorem~\ref{T:RPRAM} consists of two parts,
which relate to a new computational model introduced here, which we call the
BRP-RAM.
This model is identical to the RP-RAM model, except that the instruction
\[
x\Leftarrow \textit{RAND}(y)
\]
is now replaced by
\[
x\Leftarrow \textit{RAND}(2^k),
\]
a pseudofunction assigning to $x$ an integer random variable uniformly
distributed in the range $0\le x< 2^k$.

The first part of the proof, handled in Section~\ref{S:BRP_PEL}, proves
\begin{lem}\label{L:BRPRAM}
\[
\text{BRP-RAM}[+,[\minus],[\times],[\div],\leftarrow,[\rightarrow],\textit{Bool\/}]
=\text{PEL}.
\]
\end{lem}

The second part, handled in Section~\ref{S:RP_BRP}, then completes
the proof by establishing
\begin{lem}\label{L:generic_RAND}
\[
\text{RP-RAM}[+,[\minus],[\times],[\div],\leftarrow,[\rightarrow],\textit{Bool\/}]
=
\text{BRP-RAM}[+,[\minus],[\times],[\div],\leftarrow,[\rightarrow],\textit{Bool\/}].
\]
\end{lem}

After completing this main proof, we turn in Section~\ref{S:input} to reuse
the machinery developed in order to sharpen these results, characterising both
the power of RAMs working in any specific time complexity and the power of RAMs
working under BPP acceptance criteria.

A short conclusions section follows.

\section{Preliminaries}\label{S:preliminaries}

We begin by describing the basic tools used in the construction.

\subsection{Some redundant operations}\label{SS:redundant}

In describing the algorithms in this paper we use
``$\minus$'' and ``$\rightarrow$'', as well as all comparators, freely, even
when the only non-optional
operations available are ``$+$'', ``$\leftarrow$'' and ``$\textit{Bool\/}$''
and the only comparator available is ``$=$'', testing for equality.
In this section we justify this presentation, by showing that
``$\minus$'', ``$\rightarrow$'' and
all comparators can be simulated given the available operations.

To see this, note first that had these been RAMs working on general integers
(not necessarily
non-negative), and had numbers been stored in registers in standard
two's complement notation (see \cite{Koren:Algorithms}), then $\textit{Bool}$
would have included the standard bitwise negation operator $\lnot$, and standard
arithmetic subtraction (``$-$'') would have been implementable as
$a-b=a+\lnot b+1$. Because we deal with RAMs working over nonnegative integers,
we have a negation operator that only works up to (and including)
the most-significant ``$1$'' bit of its operand.
This means, for example, that for any $a$,
$a + \lnot a$ is a number of the form $2^m-1$, with the minimal $m$ such
that $2^m-1 \ge a$. We therefore define the function
$$\set(a)\defeq a + \lnot a.$$

To determine whether $a\le b$, for some $a$ and $b$, without utilisation of
either ``$\minus$'' or any comparison operator other than equality, consider two
cases. First, it may be that $a$ and $b$ differ in their bit-length. In this
case, it is enough to check that $\set(a)<\set(b)$, which can be implemented
by $\set(a)\lor\set(b)\neq\set(a)$. Alternatively, it may be that
$\set(a)=\set(b)$. In this case,
it is known that $a$ and $\lnot b$ are both in the range $[0,\set(b)]$,
so the question whether $a\le b$ can
be formulated equivalently as whether $a+\lnot b$ (which by definition
of $\set$ also equals
$\set(b)+(a-b)$) is in the range $[0,\set(b)]$ or in the range
$[\set(b)+1,2\set(b)]$. These two ranges can be differentiated by their value
in a single bit position, namely that indicated by $\set(b)+1$, so, altogether,
we can define
\begin{align*}
a\le b &\defeq \set(a)\lor\set(b)\neq\set(a) \\
& \text{or }(\set(a)
=\set(b)\text{ and }(a+\lnot b)\land(\set(b)+1)=0).
\end{align*}
All other comparison operators can now be computed from ``$\le$'', making their
inclusion in the RAM's instruction set superfluous.

Consider now the case $b<a$. Let $c=b+\set(a)$. Once again, by definition
$a+\lnot c=\set(c)+(a-c)=(\set(c)-\set(a))+(a-b)$. If we wish to calculate
$a\minus b$, which, in this case, equals $a-b$, we can now utilise the fact that
the result must be in the range $[0,\set(a)]$. Because $c\ge a$, we
know that $(\set(c)-\set(a))\land\set(a)=0$, so $a\minus b$ can be calculated as
$(a+\lnot c)\land \set(a)$.

Putting it all together, we can define a function to calculate $a\minus b$ as
follows.
\begin{equation*}
a \minus b\defeq
\begin{cases}0 & \text{if $a\le b$,} \\
(a+\lnot (b+\set(a))) \land \set(a) & \text{otherwise.}
\end{cases}
\end{equation*}

Right shifting is not as universally replaceable. However, it can be
simulated under appropriate circumstances, as outlined by
Lemma~\ref{L:rightshift}.

\begin{lem}\label{L:rightshift}
For
$\textit{op}=\{\leftarrow,\rightarrow,[+],[\minus],\textit{Bool\/}\}$,
if a
$\text{RAM}[\textit{op}]$ does not use indirect addressing and is restricted
to shifts by bounded amounts, it can be simulated by a
$\text{RAM}[\textit{op}\setminus\{\rightarrow\}]$ without loss in time
complexity. This result remains true also if the RAM can apply
``$a \rightarrow b$'' when $b$ is the (unbounded) contents of a register,
provided that the calculation of $b$ does not involve use of the
``$\rightarrow$'' operator.
\end{lem}

\begin{proof}
We begin by considering the case of bounded shifts.

A RAM that does not use indirect addressing is inherently able to access only
a finite set of registers. Without loss of generality, let us assume that these
are $R[0],\ldots,R[k]$. The simulating RAM will have
$R'[0],\ldots,R'[k+1]$ satisfying the invariant
\[\forall i:0\le i\le k, R[i]=R'[i]/R'[k+1].\]
To do this, we initialise $R'[k+1]$ to be $1$, and proceed with the simulation
by translating any action by the simulated RAM on $R[i]$, for any $i$, to the
same action on $R'[i]$.\footnote{An action involving an explicit ``$1$'' (except
for shifting by 1) will
have the ``$1$'' replaced by $R'[k+1]$ in the simulation.}
We do this for all actions except $R[i]\rightarrow X$,
which is an operation that is unavailable to the simulating RAM.

To simulate ``$R[j]\Leftarrow R[i]\rightarrow X$'', we perform the following.
\begin{enumerate}
\item $R'[j]\Leftarrow R'[i]$.
\item $\forall x:x\neq j, R'[x]\Leftarrow R'[x]\leftarrow X$.
\item $R'[j]\Leftarrow R'[j] \clr \lnot R'[k+1]$.
\end{enumerate}
We note regarding the second step that this operation is performed also for
$x=k+1$. The fact that $k$ is bounded ensures that this step is performed in
$\BigO(1)$ time.

Essentially, if ``$R[j]\Leftarrow R[i]\rightarrow X$'' is thought of as
``$R[j]\Leftarrow \lfloor R[i]/2^X\rfloor$'', Step 1 performs the
assignment, Step 2 the division, and Step 3 the truncation.

In order to support ``$R[j]\Leftarrow R[i]\rightarrow X$'' also when $X$ is
the product of a calculation, the simulating RAM also performs, in parallel to
all of the above, a direct
simulation that keeps track of the register's native values. In this alternate
simulation, right shifts are merely ignored. Any calculation performed by the
simulated RAM that does not involve right shifts will, however, be calculated
correctly, so the value of $X$ will always be correct.
\end{proof}

Because the conditions of Lemma~\ref{L:rightshift} hold for all constructions
described here, we assume the availability of ``$\rightarrow$'', throughout.

\subsection{Vectors and tableaux}\label{SS:known}

Consider the following definitions, following \cite{Brand:indirect}:

\begin{defi}[Vectors]\label{D:vectors}
A triplet $(m,V,n)$ of integers will be called an \emph{encoded vector}.
We refer to $m$ as the $width$ of the vector, $V$ as the \emph{contents} of the
vector and
$n$ as the $length$ of the vector. If $V=\sum_{i=0}^{n-1} 2^{mi} k_i$ with
$\forall i:0 \le i < n \Rightarrow 0 \le k_i < 2^m$, then $[k_0,\ldots, k_{n-1}]$ will be
called the \emph{vector} (or, the \emph{decoded vector}), and the $k_i$ will be
termed the \emph{vector elements}. Notably, vector elements belong to a finite
set of size $2^m$ and are not general integers. It is well-defined to consider
the most-significant bits (MSBs) of vector elements. Nevertheless, any $n$
integers can be encoded as a vector, by choosing a large enough $m$.

Actions described as operating on the vector are mathematical operations on
the encoded vector (typically, on the vector contents, $V$). However, many
times we will be more interested in analysing these mathematical operations in
terms of the effects they have on the vector elements.
Where this is not ambiguous, we will name vectors by their contents. For
example, we can talk about the ``decoded $V$'' to denote the decoded vector
corresponding to some encoded vector whose contents are $V$.
\end{defi}

\begin{defi}[Instantaneous Description]\label{D:instant_desc}
Let $\mathcal{T}$ be a TM working on a bounded tape of size $s$ and having a
state space
that can be described by $c$ bits.

An \emph{instantaneous description} of $\mathcal{T}$ at any point in its
execution is the tuple
\[
(\textit{tape},2^{\textit{headpos}},\textit{state}\times 2^{\textit{headpos}},
1,2^s)
\]
at that
point in its execution, where $\textit{tape}$ is the instantaneous content of
the TM's tape, translated into an integer,
$\textit{headpos}$ is the distance from the position of the head to the end
of the tape, and
$\textit{state}$ is an integer
indicating the current state of the finite control, where a mapping from
states to integers can be arbitrary subject to the requirement that all integers
must be at most $c$ bits in length. By convention, the initial state of the
TM is mapped to zero.

(The constants $1$ and $2^s$ are needed because they denote the endpoints of
the portion of the tape that can be occupied by the TM's head.)
\end{defi}

This particular format to describe the instantaneous state of a TM was chosen
because advancing the TM to the next instantaneous state can be done using it
in constant time with only bit shifts and bitwise Boolean operations
(but without the tweaked ``$\lnot$''). Furthermore, only
bit shifts by a constant amount are required. That is, one can avoid all
bit shift operations in which the right-hand operand depends on
$\textit{tape}$, $\textit{headpos}$, $\textit{state}$ or $s$. No constants
are required in the calculation.

Let $X_i$ (a tuple of $5$ integers) denote the instantaneous description of
$\mathcal{T}$ after $i$ execution steps. We denote by $\hat{\mathcal{T}}$ the
function that calculates in constant time
$X_{i+1}=\hat{\mathcal{T}}(X_i)$ using only bitwise
Boolean operations and shifts by a constant amount.

\begin{defi}[Tableau]
Let $\mathcal{T}$, $s$ and $c$ be as in Definition~\ref{D:instant_desc}.

A \emph{tableau} is a tuple $(m,T,H,S,n)$, where $m=s+c-1$ and the vectors
$(m,T,n)$, $(m,H,n)$ and $(m,S,n)$ are such that $(T[i],H[i],S[i],1,2^s)$
is the instantaneous description of
$\mathcal{T}$ after $i$ execution steps.
\end{defi}

A technique used extensively in \cite{Simon:division2, Brand:ALNs} to verify
the validity of an entire tableau in constant time is to define vector
$(m,I,n)$,
such that $I[i]=1$ for every $i$, as well as the scalar $I_s=I\leftarrow s$,
and then to calculate
\[
(T',H',S',I',I'_s) \Leftarrow \hat{\mathcal{T}}(T,H,S,I,I_s).
\]

With a little care, the fact that $\hat{\mathcal{T}}$ works only with bitwise
Boolean functions and constant shifts can be leveraged to make it work
independently for each instantaneous description inside the vector contents
that comprise the tableau candidate. Specifically, if $(m,T,H,S,n)$ is a valid
tableau for $\mathcal{T}$ working on input $\textit{inp}$, then for every $i$
except $i=n-1$ we should have
$T'[i]=T[i+1]$, $H'[i]=H[i+1]$, $S'[i]=S[i+1]$.

These equalities, in turn, can be verified by confirming
$((T'\leftarrow m)\oplus T)\land (1\leftarrow (nm)\minus 1)=\textit{inp}$,
$((H'\leftarrow m)\oplus H)\land (1\leftarrow (nm)\minus 1)=1$ and
$((S'\leftarrow m)\oplus S)\land (1\leftarrow (nm)\minus 1)=0$,
where ``$\oplus$'' signifies the exclusive OR (XOR) bitwise operation.
To see this, consider, for example,
$(T'\leftarrow m)\oplus T$.
In a correct tableau, this will equal $T[0]+(T[n-1]\leftarrow (nm))$. The later
``$\land$'' takes care of the part dependent on $T[n-1]$. The tableau is valid
if it starts with the correct element, $T[0]=\textit{inp}$,
and every subsequent element satisfies the relation $T'[i]=T[i+1]$. The
same verification can now be repeated for $H$ and $S$.

It is also possible to verify the correctness of $I$ and $I_s$ in a similar way:
$(I\leftarrow m)\oplus I = 1\leftarrow (nm)+1$, $I_s=I\leftarrow s$.

In practice we will always pick $n$ to be a known power of two, $n=2^k$, so
that $nm$ can be calculated as $m\leftarrow k$, with no need for a
multiplication operation.

Once a tableau is verified, examining $(T[n-1],H[n-1],S[n-1])$ can answer the
question of how the execution described by the tableau terminates. Options are
\begin{enumerate}
\item Termination and acceptance of the input.
\item Termination and rejection of the input.
\item Nontermination.
\item The TM has exceeded its allotted tape requirements.
\end{enumerate}
In constant time, we are able to tell these options apart.

Let $\hat{\mathcal{R}}$ be the constant time RAM program that verifies an input
tableau and returns its validity and termination status, as described above.
The function $\hat{\mathcal{R}}$ takes as inputs,
in addition to the original tableau and $s$, also the helper vector $I$.
The vector $I_s$ does not need to be given as an extra input, because it can
be calculated as $I\leftarrow s$.

In \cite{Brand:ALNs}, the point is made that the total number of possible
instantaneous descriptions is bounded by $B\defeq 2^{3(s+c-1)}$. Thus, there is
no need to test the validity of any tableau longer than $B$. If a TM working on
a tape of size $s$ has not terminated on its first $B$ steps, it is guaranteed
to be in an infinite loop.

The paper then goes on to make the point that this limit makes the grand total
of possible tableaux one needs to test also finite. It is the set of all
possible tableaux of length $B$.

\subsection{Maps}\label{SS:novel}

\begin{defi}[Map]
A \emph{map} (or, an \emph{encoded map}) is a triplet, $(\hat{L},\hat{I},\hat{w})$, of integers,
where $\hat{L}$ is the \emph{contents} of the map, $\hat{I}$ is the \emph{domain} of the
map and $\hat{w}$ is the \emph{width} of the map. The domain, $\hat{I}$, must satisfy the
criterion that any two ``$1$''s in its binary representation are at least $\hat{w}$
bit positions apart. The contents, $\hat{L}$, must satisfy that all
bit-positions of $\hat{L}$ are zero, except the $\hat{w}$ positions immediately
following a position that is a ``$1$'' in the domain.\footnote{The latter
criterion can be worded in a formula as
$\hat{L}=\hat{L} \land ((\hat{I} \leftarrow \hat{w})-\hat{I})$.}
For simplicity we assume $\hat{w}>1$.

Conceptually, a map represents a mapping from a list of \emph{indices}
$\textit{ind}_1,\ldots$ to a list of \emph{map elements} $k_1,\ldots$, in the following
way: $\hat{I}=\sum_i (1\leftarrow \textit{ind}_i)$ and
$\hat{L}=\sum_i (k_i \leftarrow \textit{ind}_i)$. The $k_i$ are required to be in the
range $0\le k_i < 2^{\hat{w}}$. The $\textit{ind}_i$ are required to satisfy that for
every $i$, $\hat{w}+\textit{ind}_i\le\textit{ind}_{i+1}$. Under these restrictions,
there is a one-to-one correspondence between mappings and encoded maps. Such
mappings will be termed \emph{decoded maps}. A map will often be
referred to by the name of the encoded map contents. This is done, for example,
in the following notation, which we introduce to signify the relation between an
encoded map and its underlying mapping: $k_i=\hat{L}[\textit{ind}_i]$.

Note that the map elements belong to a finite set, the elements of which are
representable by $\hat{w}$-bit strings. As such, it is well-defined to consider their
most significant bits (MSBs).

Operations on maps are RAM operations performed on the encoded maps. We will,
however, mostly be interested in the effects these have on the map elements.
Such operations may include multiple maps that share the same $\hat{I}$ and $\hat{w}$,
in which case, when we say that we apply a function $f$ on two maps, $V$ and
$U$, we typically mean that we wish to produce a result, $W$, such that
for all indices $W[\textit{ind}_i]=f(V[\textit{ind}_i],U[\textit{ind}_i])$.
\end{defi}

In Section~\ref{SS:known}, we have shown that a suitably-designed function,
$\hat{\mathcal{T}}$ that was designed to work on integers can be run with
vector contents as input, with the effect that it will function independently
on each vector element. The function $\hat{\mathcal{R}}$, also introduced in
the same section, has similar properties and can be given as input a tuple of
map contents, and will act separately on each element of the maps. The
effect is that if each map element houses a tableau candidate,
function $\hat{\mathcal{R}}$ can determine, in constant time, which of
the candidates is a valid tableau and which not, with valid tableaux being
signified by zeroes in the outputs and non-valid tableaux being signified
by nonzeros.

We make use of several variants of this parallelised
$\hat{\mathcal{R}}$. The one described
above we denote
$\hat{\mathcal{R}}_\textit{valid}$, because a zero in the output indicates
that the tableau is valid. Other variants used here are
$\hat{\mathcal{R}}_\textit{valid-and-accepting}$,
$\hat{\mathcal{R}}_\textit{valid-and-rejecting}$ and
$\hat{\mathcal{R}}_\textit{valid-and-exceeded-tape}$, which are RAM
programs in which an output of zero indicates that the input tableau is
``valid and accepting'', ``valid and rejecting'' or ``valid and indicating that
the simulated TM exceeded its tape allocation'', respectively.
Each of these can be programmed by applying the same techniques as shown here.

The most salient difference between the parallelised version of
$\hat{\mathcal{R}}$ and its standard version is that the standard version
used freely constants such as $1$ and $\textit{inp}$ in calculating and
verifying $T[0]$, $H[0]$ and $S[0]$. For these to be usable by the parallelised
version of $\hat{\mathcal{R}}$, they have to
be multiplied by the map domain, $\hat{I}$.
The constant $1$ becomes $1\times \hat{I}=\hat{I}$, which is available already,
but
calculating $\textit{inp}\times \hat{I}$ may not be as straightforward. Whereas
the rest of the computations described here can be performed in constant time,
calculating $\textit{IN}=\textit{inp}\times \hat{I}$ requires $\Th(n)$ time,
where
$n$ is the bit-length of $\textit{inp}$. The calculation algorithm itself is
the classic long-multiplication algorithm. It is described in
Algorithm~\ref{A:mult}.

\begin{algorithm}
\caption{Calculating $\textit{IN}=\hat{I}\times \textit{inp}$}
\label{A:mult}
\begin{algorithmic}[1]
\State $\textit{IN} \Leftarrow 0$
\State $\textit{acc} \Leftarrow \hat{I}$
\State $r \Leftarrow \textit{inp}$
\While {$r\neq 0$}
\If {$r\land 1=1$}
\State $\textit{IN} \Leftarrow \textit{IN}+\textit{acc}$
\EndIf
\State $r\rightarrow 1$
\State $\textit{acc}\leftarrow 1$
\EndWhile
\State \Return $\textit{IN}$
\end{algorithmic}
\end{algorithm}

The tool needed in order to work conveniently with such an output is an
element-by-element comparator. We would like to check for each element
equality to zero and signify this by a ``$1$'', with non-equality being
signified by zero. A sequence of function definitions working on maps,
culminating in a general element-by-element equality operation that works
in constant time, parallelising over all map elements, is described below.

Let $(V,I,w)$ and $(U,I,w)$ be two maps.

The total set of bits that can be used by the map contents is
\[
\textit{MASK}(I,w)\defeq(I\leftarrow w)\minus I.
\]

For convenience, we divide these into two subsets. We take the lowest $w-1$
bits of each map element to be its
``data'' bits, and the MSB to be the ``flag'' bit. The following
functions extract these positions.
\[
\textit{FLAGS}(I,w)\defeq I\leftarrow (w\minus 1),
\]
\[
\textit{DATA}(I,w)\defeq\textit{FLAGS}(I,w)\minus I.
\]

To actually extract the data from the positions we use
\[
\textit{FLAGS}(V,I,w)\defeq V\land \textit{FLAGS}(I,w),
\]
\[
\textit{DATA}(V,I,w)\defeq V\land \textit{DATA}(I,w).
\]

Summing the data of two maps can be performed by
\[
\textit{ADD}(V,U,I,w)\defeq(\textit{DATA}(V,I,w)+\textit{DATA}(U,I,w))\oplus \textit{FLAGS}(V,I,w)\oplus \textit{FLAGS}(U,I,w).
\]
The ADD function, as implemented here, avoids overflows by calculating the sum
modulo $2^w$. In order to
find out if an overflow occurred, we can calculate the carry bit.
\begin{align*}
\textit{CARRY}(V,U,I,w) \defeq
(V+U\minus \textit{ADD}(V,U,I,w))\rightarrow w.
\end{align*}

The following function implements bitwise negation:
\[
\textit{NEG}(V,I,w)\defeq\textit{MASK}(I,w)\minus V,
\]

and, as a last helper function, the following function calculates
a map, $\textit{RC}$, such that for every $i$,
$\textit{RC}[\textit{ind}_i]$ is $1$ if $V[\textit{ind}_i]>U[\textit{ind}_i]$,
and is $0$, otherwise, this being an element-by-element ``greater than''
comparison:
\[
\textit{GT}(V,U,I,w)\defeq\textit{CARRY}(V,\textit{NEG}(U,I,w),I,w).
\]

With this build-up, we can finally implement element-by-element equality as
follows:
\[
\textit{EQ}(V,U,I,w)\defeq I\minus\textit{GT}(V,U,I,w)\minus \textit{GT}(U,V,I,w).
\]

In order to simultaneously check and verify many candidate tableaux, what we
need are integers $\hat{L}$ and $\hat{I}$, as well as
the length, $s$, of the tape on which TM $\mathcal{T}$ runs, and the TM's input,
$\textit{inp}$. The inputs $\hat{L}$ and $\hat{I}$ can
be interpreted as follows.

Let $w=s+c\minus 1$ be the width of the
tableau vectors being checked, $B=1\leftarrow (w+w+w)$ will be the length of the
tableau vectors being checked. The total bit-length of these vectors is
$w_m=B\times w=w\leftarrow (w+w+w)$. Therefore, they can be stored as
elements inside maps of such width. Specifically, let $\hat{w}=4w_m$.
The map $(\hat{L}, \hat{I}, \hat{w})$ holds the information of all tableaux to
be verified, by keeping in each element, $i$, of the map the following
composite:
\[
L[i]=(T_i\leftarrow 3w_m)+(H_i\leftarrow 2w_m)+(S_i\leftarrow w_m)+I_i,
\]
where $(w,T_i,H_i,S_i,B)$ is the $i$'th tableau to be verified, and
$(w,I_i,B)$ is the additional vector (the vector whose elements are all $1$)
required to execute $\hat{\mathcal{R}}$.

In order to perform the actual verification, we begin by calculating
$w_m$, and then
\begin{align*}
T &\Leftarrow (\hat{L}\rightarrow 3w_m)\land\textit{MASK}(w_m) \\
H &\Leftarrow (\hat{L}\rightarrow 2w_m)\land\textit{MASK}(w_m) \\
S &\Leftarrow (\hat{L}\rightarrow w_m)\land\textit{MASK}(w_m) \\
I &\Leftarrow \hat{L}\land\textit{MASK}(w_m),
\end{align*}
this unpacking all $T_i$, $H_i$, $S_i$ and $I_i$ simultaneously as the $i$'th
elements of the maps
$(T,\hat{I},w_m)$, $(H,\hat{I},w_m)$, $(S,\hat{I},w_m)$ and
$(I,\hat{I},w_m)$, respectively.

As a last preparation step, we calculate $\textit{IN}=\hat{I}\times\textit{inp}$
by means of Algorithm~\ref{A:mult}.

We are now ready to run $\hat{\mathcal{R}}_\textit{valid}$ on the inputs
$(T,H,S,I,s)$. This is essentially the same program as $\hat{\mathcal{R}}$, with
the only differences being
that ``$1$''s in $\hat{\mathcal{R}}$'s code are are replaced by ``$\hat{I}$''s
in $\hat{\mathcal{R}}_\textit{valid}$ (except when shifting by $1$) and
``$\textit{inp}$''s are replaced by ``$\textit{IN}$''s.
The program's output is a map, $(R,\hat{I},w_m)$, such that
$R[i]$ is zero if and only if the $i$'th element of the
input is a valid tableau candidate, and nonzero otherwise.

By running $\textit{EQ}$ over this vector, to compare it with the zero vector,
we reverse the zero/nonzero distinction. Now, nonzero (one) elements represent
valid tableau candidates and zero elements represent invalid candidates.
Consider the contents of this map. If it is nonzero, this should be interpreted
as ``some of the tableau candidates in the input are valid''.

We can, similarly, create such vectors for valid-and-accepting tableaux and
valid-and-rejecting tableaux, etc..
Here, the meaning is ``some of the tableau candidates are valid and accepting''
(or valid and rejecting), from which we can conclude that $\mathcal{T}$
accepts (rejects).

The full verification algorithm is given as Algorithm~\ref{A:map_verify}.

\begin{algorithm}
\caption{$\textit{Verify}(\hat{L},\hat{I},s,\textit{inp})$: Simultaneous verification of tableau candidates}
\label{A:map_verify}
\begin{algorithmic}[1]
\State $w\Leftarrow s+c\minus 1$
\State $B\Leftarrow 1\leftarrow (w+w+w)$
\State $w_m\Leftarrow w\leftarrow (w+w+w)$
\State $T \Leftarrow (L\rightarrow (w_m+w_m+w_m))\land\textit{MASK}(w_m)$
\State $H \Leftarrow (L\rightarrow (w_m+w_m))\land\textit{MASK}(w_m)$
\State $S \Leftarrow (L\rightarrow w_m)\land\textit{MASK}(w_m)$
\State $I \Leftarrow L\land\textit{MASK}(w_m)$
\State $\textit{IN} \Leftarrow \hat{I}\times\textit{inp}$
\Comment Multiplication performed using Algorithm~\ref{A:mult}.
\State $R_\textit{valid} \Leftarrow \hat{\mathcal{R}}_\textit{valid}(T,H,S,I,\hat{I},\textit{IN},s)$
\State $R_\textit{valid-and-accepting} \Leftarrow \hat{\mathcal{R}}_\textit{valid-and-accepting}(T,H,S,I,\hat{I},\textit{IN},s)$
\State $R_\textit{valid-and-rejecting} \Leftarrow \hat{\mathcal{R}}_\textit{valid-and-rejecting}(T,H,S,I,\hat{I},\textit{IN},s)$
\State $R_\textit{valid-and-exceeded-tape} \Leftarrow \hat{\mathcal{R}}_\textit{valid-and-exceeded-tape}(T,H,S,I,\hat{I},\textit{IN},s)$
\If {$\textit{EQ}(R_\textit{valid-and-accepting},0,\hat{I},w_m)\ne 0$}
\State \textbf{Output:} $\mathcal{T}$ accepts on input $\textit{inp}$.
\ElsIf {$\textit{EQ}(R_\textit{valid-and-rejecting},0,\hat{I},w_m)\ne 0$}
\State \textbf{Output:} $\mathcal{T}$ rejects on input $\textit{inp}$.
\ElsIf {$\textit{EQ}(R_\textit{valid-and-exceeded-tape},0,\hat{I},w_m)\ne 0$}
\State \textbf{Output:} $\mathcal{T}$ requires more than $s$ tape elements on input $\textit{inp}$.
\ElsIf {$\textit{EQ}(R_\textit{valid},0,\hat{I},w_m)\ne 0$}
\State \textbf{Output:} $\mathcal{T}$ rejects input $\textit{inp}$ by entering an infinite loop.
\Else
\State \textbf{Output:} Simulation failed. No valid tableau found. It was not determined whether $\mathcal{T}$ accepts input $\textit{inp}$.
\EndIf
\end{algorithmic}
\end{algorithm}

\section{$\text{BRP-RAM}=\text{PEL}$}\label{S:BRP_PEL}

Arguably, the more difficult direction in proving Lemma~\ref{L:BRPRAM} is
\[
\text{PEL}\subseteq\text{BRP-RAM}[+,\leftarrow,\textit{Bool\/}].
\]
We prove this by means of an explicit algorithm that simulates any TM running
in space PEL, on a BRP-RAM. The top level of this algorithm is depicted
in Algorithm~\ref{A:R_sim_TM}. This algorithm simulates a TM, $\mathcal{T}$,
the state of whose finite control can be described in $c$ bits, running on
input $\textit{inp}$.

\begin{algorithm}
\caption{Top level of a BRP-RAM simulating a PEL TM, $\mathcal{T}$, working on input $\textit{inp}$.}
\label{A:R_sim_TM}
\begin{algorithmic}[1]
\State $\textit{maxstep} \Leftarrow 1$
\Loop \label{Step:infiloop}
\State $s\Leftarrow {}^\textit{maxstep}2$\label{Step:s}
\State Generate $\hat{L}$ and $\hat{I}$ for simulation on tape of length $s$.\label{Step:LI}
\Comment See Section~\ref{SS:mult}.
\State $\textit{Verify}(\hat{L},\hat{I},s,\textit{inp})$\label{Step:sim}
\Comment Probabilistic verification. See Algorithm~\ref{A:map_verify}.
\State Check if the simulation succeeded. Halt and reject if not.\label{Step:check_sim}
\State Check if the simulated $\mathcal{T}$ has exceeded its tape allocation. Halt and report the final state if not. (If it has not reached a final state, reject.)\label{Step:check_fin}
\State $\textit{maxstep} \Leftarrow \textit{maxstep}+\textit{maxstep}$
\EndLoop
\end{algorithmic}
\end{algorithm}

We require Algorithm~\ref{A:R_sim_TM} to be a BRP-RAM algorithm.
This can be broken down into three conditions.
\begin{enumerate}
\item If the run of $\mathcal{T}$ has tape requirements bounded by PEL, the
algorithm must run, deterministically, in polytime.
\item If the input is not in the language accepted by $\mathcal{T}$, it must be
rejected.
\item If the input is in the language, it must be accepted with probability at least $1/2$.
\end{enumerate}

Let us begin by considering the first condition.
Algorithm~\ref{A:R_sim_TM} contains a loop on
Step~\ref{Step:infiloop}. The complexity of the algorithm is determined by
the number of times the algorithm goes through the loop before halting at
either Step~\ref{Step:check_sim} or Step~\ref{Step:check_fin}, and by
the complexities of the individual steps.

In Algorithm~\ref{A:R_sim_TM}, Step~\ref{Step:s} can be implemented by the
straightforward algorithm, which acts in $\BigO(\textit{maxstep})$ time.
Step~\ref{Step:sim} acts in $\BigO(n)$ time, where $n$ is the bit-length of
the output. This time requirement stems from the application of
Algorithm~\ref{A:mult}. All other parts of this step run in constant time.
The greatest challenge in constructing the
algorithm, to which Sections~\ref{SS:mult} and \ref{SS:k} are dedicated, is to
implement Step~\ref{Step:LI} efficiently. Specifically, we implement it by use
of an $\BigO(\textit{maxstep})$ time algorithm.
The rest of the body of the loop of Step~\ref{Step:infiloop} runs in constant
time, for a total of $\BigO(\textit{maxstep}+n)$ time execution.

Ultimately, the combined complexity of the algorithm is $\Th(m+n\log m)$,
where $m$ is the ultimate value of $\textit{maxstep}$ in running the
algorithm,
which is also $\Th(m' +n \log m')$, where $m'$ is the penultimate value of
$\textit{maxstep}$. Because we know that the algorithm did not halt during its
penultimate cycle through the loop, we know that the simulation of $\mathcal{T}$
over a tape of length $s'={}^{m'}2$ was successful
(Step~\ref{Step:check_sim}) and showed that $\mathcal{T}$ requires more than
$s'$
tape elements (Step~\ref{Step:check_fin}).

Because the space requirement for $\mathcal{T}$ is known to
be more than $s'$, and because by assumption it is in PEL,
by definition $m'$ is a polynomial in $n$. Hence, the first
condition, that of the RAM running in polynomial time, is
satisfied.

The second condition is that all inputs not in the language must be rejected.
The only way to accept an input in Algorithm~\ref{A:R_sim_TM} is in
Step~\ref{Step:check_fin}, when the simulation shows that $\mathcal{T}$ has
halted and accepted the input, after the simulation was verified as correct in
Step~\ref{Step:check_sim}. Thus, the second condition is also satisfied.

The third condition is that if $\textit{inp}$ is in the language, its
probability of being rejected should be at most $1/2$. The only place where
Algorithm~\ref{A:R_sim_TM} can falsely reject is at Step~\ref{Step:check_sim}.
If the simulation is correct, Step~\ref{Step:check_fin} is able to determine
$\mathcal{T}$'s true final state. It neither falsely rejects nor falsely
accepts.

Suppose now that we are able to devise our simulation of $\mathcal{T}$ so
that in the first iteration through the loop it has probability $p_1$ of
failing, in the second iteration it has probability $p_2$, etc., such that
the sum of the entire $p_i$ sequence is no more than $1/2$. (For example,
we can set $p_i=2^{-1-i}$.) The probability that at least one failure
occurred is certainly no more than the sum of all $p_i$, indicating that the
probability of false rejection is properly bounded.

Sections~\ref{SS:mult} and \ref{SS:k} describe how to generate
$\hat{L}$ and $\hat{I}$ in Step~\ref{Step:LI} of the algorithm, so as to meet
the requirement that a failure of the algorithm occurs in probability bounded
by $2^{-1-i}$, thus completing the description of the simulating algorithm.

\subsection{Calculating $\hat{L}$ and $\hat{I}$}\label{SS:mult}

Summarising the conditions required of the map
$(\hat{L},\hat{I}, \hat{w} )$ to be generated at iteration
$i$ of the loop on Step~\ref{Step:infiloop} of Algorithm~\ref{A:R_sim_TM}:
\begin{enumerate}
\item In the binary representation of the output $\hat{I}$, any two ``$1$''s should be
at least $ \hat{w} $ bit positions apart.
(We refer to this condition as $ \hat{w} $-sparseness.)
\item The correct tableau, which is a bit-string of length $ \hat{w} $ bits, should appear with probability of at least $1-2^{-i-1}$ as a substring of the bits of $\hat{L}$ beginning at some ``1'' position of $\hat{I}$.
\item All bits of $\hat{L}$ which are more than $ \hat{w} $ bit positions from the preceding ``1'' position of $\hat{I}$ should be zeroes.
\item The generation of $I$ and $L$, with
$\hat{w}$ in ${}^{\Th(\textit{maxstep})}2$, should be performed in
$\BigO(\textit{maxstep})$ time.
\end{enumerate}

The general framework we use to meet these requirements is as follows.
First, we choose a number, $k$, as a function of $\textit{maxstep}$.
Then, we generate $\hat{I}$ to be a valid map domain in the range $[0,2^{k- \hat{w} })$.
Last, we generate $\hat{L}$ to meet with the remaining requirements.

Generating $\hat{L}$ is ultimately done by
\[
\hat{L}\Leftarrow \textit{RAND}(2^k)\land\textit{MASK}(\hat{I},\hat{w}).
\]
This constant time procedure ensures that if $\hat{I}$ is a valid map domain
of Hamming weight
$h$, the resulting map will be valid, and will include $h$ independent,
uniformly chosen elements. The correct tableau is one of the possible elements
for the map. Thus, the greater $h$ the higher the probability that one of the
randomly-chosen elements is the correct one.
We will design $\hat{I}$ and $k$ so that $h$ will be,
probabilistically, high enough to ensure that this probability is at least
$1-p_i>1-2^{-i-1}$.

In this section we discuss the generation of $\hat{I}$ given a choice of
$k$, and in Section~\ref{SS:k} the generation of $k$.

Consider the sparseness condition for $\hat{I}$.
We note that unlike $\hat{I}$'s Hamming weight, which may be probabilistically
chosen to meet with the probabilistic success criteria, the condition
of $ \hat{w} $-sparseness is a deterministic condition.
We cannot generate $I$ by a random process that is simply biased towards ``$0$''
bits. Even if such a process has high probability of meeting the criterion,
it still admits the possibility that the
condition will not be satisfied and, as a result, that the calculation will be
incorrect. Such an approach can create a positive false-accept probability,
which is not acceptable in the BRP model.

We construct $\hat{I}$ by means of a procedure that iteratively dilutes the
``1''s in $\hat{I}$. After $i$ steps, $\hat{I}$ is guaranteed to be
$w_i$-sparse, with $w_i> {}^i 2$. In this way, generating a
$ \hat{w} $-sparse $\hat{I}$ can be done in $\BigO(\textit{maxstep})$ steps.
Each step will be accomplished by a
constant time procedure, so the total time complexity of
the algorithm is as required.

Let $R_i$, for $i=0,\ldots$, be independent random values generated by calls
to $\textit{RAND}(2^k)$ for our chosen $k$.
In the definition of maps we assumed that the width of any map is at least $2$,
so we bootstrap the dilution process by Algorithm~\ref{A:first_step}.

\begin{algorithm}
\caption{Creating $\hat{I}$ for $w_0=2$}
\label{A:first_step}
\begin{algorithmic}[1]
\State $\hat{I}\Leftarrow R_0$
\State $\hat{I}\Leftarrow \hat{I} \clr (\hat{I}\leftarrow 1)$
\State $\hat{I}\Leftarrow \hat{I}\rightarrow 1$
\State \Return $\hat{I}$
\end{algorithmic}
\end{algorithm}

This procedure generates a $w_0$-sparse candidate for $\hat{I}$. We label it
$I_0$.
From this point on we iteratively begin, in step $i+1$, with a pair
$(I_i,w_i)$, where $I_i$ is a candidate for $\hat{I}$ satisfying
$w_i$-sparseness, and generate $(I_{i+1},w_{i+1})$ for the next iteration.
Algorithm~\ref{A:induction_step} describes how to do this with
$w_{i+1}=(w_i \leftarrow w_i)+1$, which is a growth rate that meets with the
requirements of the algorithm.

\begin{algorithm}
\caption{Creating $(I_{i+1},w_{i+1})$ from $(I_i,w_i)$}
\label{A:induction_step}
\begin{algorithmic}[1]
\State $R_i\Leftarrow \textit{RAND}(2^k)$
\State $I_\textit{begin} \Leftarrow I_i\clr (I_i \leftarrow w_i)$
\State $I_\textit{end} \Leftarrow (I_i\leftarrow w_i)\clr I_i$
\State $I_\textit{middle} \Leftarrow I_i \minus I_\textit{begin}$
\State $I_\textit{goodbegin} \Leftarrow \textit{EQ}(R_i \land \textit{MASK}(I_\textit{begin},w_i),0,I_\textit{begin},w_i)$
\State $I_\textit{goodend} \Leftarrow \textit{EQ}((R_i \leftarrow w_i) \land \textit{MASK}(I_\textit{end},w_i),\textit{MASK}(I_\textit{end},w_i),I_\textit{end},w_i)$
\State $I_\textit{goodmiddle} \Leftarrow \textit{EQ}(ADD((R_i \leftarrow w_i) \land \textit{MASK}(I_\textit{middle},w_i),I_\textit{middle},I_\textit{middle},w_i),R_i \land \textit{MASK}(I_\textit{middle},w_i),I_\textit{middle},w_i)$
\State $w_{i+1} \Leftarrow (w_i\leftarrow w_i)+1$
\State $I_{i+1} \Leftarrow ((\textit{MASK}(I_\textit{goodbegin}+I_\textit{goodmiddle},w_i)+I_\textit{goodbegin}) \land I_\textit{goodend}) \rightarrow (w_{i+1}\minus 1)$
\State \Return $(I_{i+1},w_{i+1})$
\end{algorithmic}
\end{algorithm}

Algorithm~\ref{A:induction_step} is the core algorithm, at the heart of this
paper's entire construction. We therefore examine and explain it line by line.

The idea behind Algorithm~\ref{A:induction_step} is to use the same tools we
developed in order to verify a tableau, only this time to use them in order to
verify the separation between ``$1$'' bits. Specifically, we begin by assuming
that every ``$1$'' bit in $I_i$ is followed by $w_i-1$ zero bits. This entire
integer can therefore be thought of as occurrences of the repeating pattern
``$(\mathsf{0}^{w_i-1}\mathsf{1})^+$'', separated
by zeroes. We want to measure the lengths of these repeating patterns. It is
straightforward to find where such a pattern begins, where it ends, and
which ``$1$'' bits are its middle bits. These are designated $I_\textit{begin}$,
$I_\textit{end}$ and $I_\textit{middle}$, respectively. The question is only
how to count the number of consecutive ``middle'' bits.

To do this, we treat $R_i$ as an Oracle string. Specifically: a counter.
We verify (in $I_\textit{goodbegin}$, $I_\textit{goodend}$ and
$I_\textit{goodmiddle}$, respectively) that the counter begins with a zero,
ends with $2^w-1$ and increments each time by one.
In the last step, the addition by $I_\textit{goodbegin}$ sends a
carry bit through the entire verified part of the vector. If it reaches
$I_\textit{goodend}$, this is an indication that the counting was correct.
We can therefore now take one bit ($I_\textit{goodend}$) from the sequence,
knowing that there must be at least $w_i \times 2^{w_i}$ zero bits preceding it.

We remark that because all we are interested in is to verify that the ``$1$''
bits are spaced far enough apart, it is not important to check, for example,
that the counter did not go through several entire revolutions, instead of just
one, between the beginning and the ending of each repetition. The omitted
checks are all for conditions which, if invalidated, merely extend the number
of zero bits that separate the ``$1$'' bits.

\subsection{Calculating $k$}\label{SS:k}

In Section~\ref{SS:mult}, $\hat{L}$ and $\hat{I}$ are calculated based on a
chosen $k$.
We now complete their construction by determining which value of $k$ to use
in the $i$'th iteration over the loop of Step~\ref{Step:infiloop} of
Algorithm~\ref{A:R_sim_TM}.
A good value would be one for which the correct tableau, describing the run of
$\mathcal{T}$ on a tape of size $s$ for its first $B$ steps,
is an element of $(\hat{L},\hat{I}, \hat{w} )$ with probability of $1-p_i$,
where $p_i<2^{-i-1}$, and therefore $\sum_i p_i\le 1/2$.

Let $m$ be the number of iterations required by Algorithm~\ref{A:induction_step}
before reaching $w_m\ge \hat{w}$. (The value of $m$ is
$\textit{maxstep}+\BigO(1)$.)

Consider a specific bit-position of $\hat{I}$, $b$, neither close to the LSB
nor to the $k$'th bit position. For $\hat{I}$'s value in
such a bit position to be $1$, each of $R_0$ through $R_m$
should have exactly
a prescribed set of values in each of $\BigO(w_m)$ bit
positions surrounding
$b$. (If there is more than one such possibility, pick one arbitrarily.)
Given that the value of $\hat{I}$ is $1$ in this bit position, we require that
the value of $\hat{L}$ in a further $ \hat{w} $ bit positions be exactly the correct
tableau.
Altogether, we need
$Z=\BigO(w_m\times m)$ bits to be
randomly chosen to exactly the appropriate values.

Let us divide the set of bit positions of $\hat{I}$ into segments of size
$W=\BigO(w_m)$,
chosen such that there will not be any overlap between the bit positions in
$R_0,\ldots, R_m$ required for the least bit inside each segment to be $1$ in
$\hat{I}$. The probability that
this bit position of $\hat{I}$ is 1, and the associated element in $\hat{L}$ is
a correct tableau is $2^{-Z}$. Suppose we pick
$k$ to be $W\times (i+1)\times 2^Z$. This ensures that there are
$(i+1)\times 2^Z$ concurrent attempts, each of which has a
$2^{-Z}$
independent probability of providing a correct tableau. Altogether, the
probability that no such tableau exists (and the simulation will therefore
fail) is on the order of
$p_i\approx \text{e}^{-i-1}<2^{-i-1}$, as desired.

In practice, we cannot use this value for $k$ because we cannot compute it
for lack of a multiplication operation. However, we can compute values that
are guaranteed to be no smaller than it, for example by switching every
$a\times b$ in the calculation to $a\leftarrow b$. A larger value for $k$
results in a smaller probability of false acceptance.

Thus, the program meets the RP-RAM acceptance criteria.

\subsection{Completing the proof}\label{SS:complete}

\begin{proof}[Proof of Lemma~\ref{L:BRPRAM}]
The techniques developed thus far provide most of the proof of the lemma:
Algorithm~\ref{A:R_sim_TM} provides the top level of a BRP simulator for PEL,
the details of which are provided in the sections following it.

To complete the proof, we now discuss the reverse direction:
\begin{equation}\label{Eq:BRP_in_PEL}
\text{BRP-RAM}\subseteq\text{PEL}.
\end{equation}

To show this,
we first note that, just like PEL can equally be thought of as PEL-TIME or as
PEL-SPACE, so can it be thought of as PEL-NSPACE. This is a direct result of
Savitch's Theorem \cite{Savitch:nspace}. This indicates that for PEL we can
replace the deterministic TM with a nondeterministic one, without this
affecting its computational power. The power of the class of randomised TMs is
clearly sandwiched between the deterministic and the nondeterministic classes,
so it must equal both. For this reason, in demonstrating
Equation~\eqref{Eq:BRP_in_PEL}, it is enough that we show that a PEL-SPACE
randomised TM can simulate a BRP-RAM.

In order to accomplish this simulation, we retain on the TM's tape the status
of the RAM's registers encoded as address-value pairs for the registers whose
value is nonzero. The number of these is linear and the size of each cannot
exceed ${}^{\BigO(\textit{step})} \max(2,\textit{inp})$ on the $\textit{step}$'th
execution step. The total space required,
including a scratchpad area to perform the actual calculations, is therefore
no more than PEL-SPACE, as necessary.

The reason the simulation was performed
using a randomised TM is in order to simulate
\[
X\Leftarrow \textit{RAND}(2^k)
\]
instructions, which are handled by writing $k$ of the TM's random bits into
$X$.
\end{proof}

\section{RP vs.\ BRP}\label{S:RP_BRP}

The discussion above pertained to a RAM model that incorporates an
$X\Leftarrow \textit{RAND}(2^k)$ pseudofunction. This still leaves the question
of whether the more general $X\Leftarrow \textit{RAND}(Y)$ is perhaps more
powerful. We complete the proof of Theorem~\ref{T:RPRAM} by showing that this
is not the case.

\begin{proof}[Proof of Lemma~\ref{L:generic_RAND}]
We construct a BRP-RAM simulation of an RP-RAM.

First, due to Lemma~\ref{L:BRPRAM}, we know that we can assume the existence
of ``$\div$'', ``$\times$'' and ``$\minus$'', from which we can further assume
the existence of a modulo operation.

Second, we can reuse a technique that was showcased in
Algorithm~\ref{A:R_sim_TM}, wherein simulation to an unknown number of steps
is performed by a sequence of simulations to a bounded and exponentially
increasing number of steps (denoted by ``$\textit{maxstep}$'') without this
affecting run-time complexities. This method allows to convert a finite-time
simulation of an RP-RAM on a BRP-RAM to a general simulation. We can, therefore,
limit ourselves to constructing a simulation bounded by $\textit{maxstep}$
steps.

During $\textit{maxstep}$ steps, a randomised RAM can invoke
$\textit{RAND}(Y_i)$ at most $\textit{maxstep}$ times, and the $Y_i$ parameter
used in each invocation is at most
$M={}^{\BigO(\textit{maxstep})}\max(2,\textit{inp})$.

The first step in the proof is to collate all calls to $\textit{RAND}(Y_i)$ in
the $\textit{maxstep}$-step
simulation into a single call. If we were able to know in advance the $Y_i$
parameter used in each call, it would have been possible to make all calls into
a single call whose parameter is the product of all $Y$ parameters used.
The individual random results can then be separated by applications of
``$\div$'' and ``$\modop$''.

In fact, we do not know the parameters in advance, but can bound their product,
$Y=\prod Y_i$,
by $M^\textit{maxstep}$. Clearly, it is possible in polynomial time to generate
a value, $2^{\tilde{k}}$, that exceeds this limit.

Let us now continue with the result from the call to
$\textit{RAND}(2^{\tilde{k}})$ as
though it was a call to $\textit{RAND}(Y)$. The extraction process of the
individual $X_i$ from $X$ depends only on $X \modop Y$. As such,
the termination probabilities that it affords are the same as those
of the RP-RAM if we condition over $X<2^{\tilde{k}}-(2^{\tilde{k}} \modop Y)$.
For higher $X$ values, the BRP-RAM will not accept the input if it is not in the
language, and will accept the input with some probability if it is in the
language.

It is not difficult to see that the worst-case for the total probability
of acceptance for an input that is within the language is
$p/(2-p)$ for the BRP-RAM, if this probability is $p$ for the
RP-RAM.\footnote{The total RP-RAM algorithm
must succeed with probability $0.5$ or more, but this does not mean that the
same is true for each one of the bounded-step simulations individually. This is
the reason why a general parameter, $p$, is required.}
This worst-case is attained when $2^{\tilde{k}}$ is approximately $(2-p) Y$, and the
accepting computations are those that work with $X\ge Y (1-p)$.

To meet with the acceptance criteria of the BRP-RAM model, we simply run the
simulation three times for each $\textit{maxstep}$ value. If in each of three
independent runs the
probability of acceptance is at least $p/(2-p)$, the probability of acceptance
in at least one of the three is at least
\[
1-\left( 1- \frac{p}{2-p} \right)^3.
\]

A little calculus shows that this is never less than $p$ for the entire range
$0\le p \le 1$, so the success rate of the simulating BRP-RAM is always at least
as good as that of the simulated RP-RAM.

On the reverse direction, an RP-RAM can simulate a BRP-RAM trivially, because
it has ``$\leftarrow$'' as a basic operation.
\end{proof}

\section{Additional results}\label{S:input}

\setcounter{thm}{1}
\setcounter{cor}{2}

We introduce two corollaries to Theorem~\ref{T:RPRAM}. First,
Theorem~\ref{T:RPRAM} relates only to polynomial-time execution. We sharpen
this result by proving
\begin{cor}\label{C:fTIME}
For any function $f(n)$, a
$\text{RAM}[+,[\minus],[\times],[\div],\leftarrow,[\rightarrow],\textit{Bool\/}]$
with access to a $\textit{RAND}()$ instruction,
working in $\Theta(f(n))$ time, where $n$ is the bit-length of the input,
can accept a language $S$ in the RP sense
(all inputs not in the language are rejected, all inputs in the language are
accepted with probability at least $1/2$) if and only if the language can be
accepted by a Turing machine working in ${}^{\Theta(f(n))}\max(2,n)$ time.
\end{cor}

Second, we consider the more relaxed acceptance criteria (and therefore, the
ostensibly stronger computational model) of BPP. Here, we ask only that all
inputs not in the language be rejected with probability at least $2/3$ and that
all inputs in the language will be accepted with probability at least $2/3$.
We claim:
\begin{cor}\label{C:BPP}
For any function $f(n)$, a
$\text{RAM}[+,[\minus],[\times],[\div],\leftarrow,[\rightarrow],\textit{Bool\/}]$
with access to a $\textit{RAND}()$ instruction,
working in $\Theta(f(n))$ time, where $n$ is the bit-length of the input,
can accept a language $S$ in the BPP sense
if and only if the language can be
accepted by a Turing machine working in ${}^{\Theta(f(n))}\max(2,n)$ time.

In particular,
\[
\text{BPP-RAM}[+,[\minus],[\times],[\div],\leftarrow,[\rightarrow],\textit{Bool\/}]=\text{PEL},
\]
where $\text{BPP-RAM}$ denotes a RAM working in polynomial time, using BPP
acceptance criteria.
\end{cor}

\begin{proof}[Proof of Corollary~\ref{C:fTIME}]

Most parts of the construction used in the proof of
Theorem~\ref{T:RPRAM} can be used as-is in the present proof. By replacing
arbitrary polynomial bounds by concrete $f(n)$ bounds, we conclude that
a RAM machine working in $\Theta(f(n))$ time can be simulated by a Turing
machine working in ${}^{\BigO(f(n))}\max(2,n)$ time, as desired.\footnote{In the
construction of Section~\ref{SS:complete}, we use
${}^{\BigO(f(n))}\max(2,\textit{inp})$, rather than ${}^{\BigO(f(n))}\max(2,n)$, but
it is not difficult to ascertain that the two are the same. Let us take, for
simplicity, $\textit{inp}$ to be $2^n$, with $n\ge 2$. We have
${}^X \textit{inp} = {}^X (2^n) \le {}^X (n^n) \le {}^{2X} n$, which is
${}^{\Th(X)} n$, as required.}
In the opposite direction,
a Turing machine working in ${}^{\Theta(f(n))}\max(2,n)$ time can be simulated
by a RAM in $\BigO(f(n)+n \log f(n))$ time. For this, the only alteration required
in Algorithm~\ref{A:R_sim_TM} is for Step~\ref{Step:s} to be changed to
\[
s \Leftarrow {}^{\textit{maxstep}}\max(2,n).
\]

For $f(n)$ functions even as slow-growing as $n^2$, the $f(n)$ factor is the
dominant part of the complexity, meeting the conditions of the lemma. However,
when $f(n)$ is small (for example, if it were a constant), the $n \log f(n)$
becomes the dominant factor. We show that this factor can be eliminated.

The extraneous $n \log f(n)$ factor comes from
Algorithm~\ref{A:mult}, which requires $\BigO(n)$ time. We
remind that the purpose of Algorithm~\ref{A:mult} was to perform the
multiplication $\textit{IN}=\hat{I}\times\textit{inp}$. In verifying the correctness
of an accepting tableau, one needs to ascertain that it contains the correct
initialisation, a correct progression from state to state, and an accepting
final state. Without $\textit{IN}$, we are able to verify the correct
progression and the properties of the final state, but we cannot, given the
means described so far, ascertain that its initialisation matches the given
input.

In a correct initialisation we have
\begin{align*}
T[0] &= \textit{inp} \\
H[0] &= 1 \\
S[0] &= 0 \\
I[0] &= 1.
\end{align*}
The construction in the proof of Lemma~\ref{L:BRPRAM} already provides means to
verify the correctness of $H[0]$, $S[0]$ and $I[0]$ in constant time. The
difficulty is in verifying $T[0]$, for which the method of Lemma~\ref{L:BRPRAM}
requires $\textit{IN}=\hat{I}\times\textit{inp}$.

Where the present construction differs from the original construction is that
instead of \emph{verifying} $T[0]$ for a
randomly generated tableau, we actively set $T[0]$ to its desired
value. We will do this simultaneously in all tableau candidates in the entire
$(\hat{L},\hat{I}, \hat{w} )$ map.

We begin by generalising somewhat the notion of a tableau.
Specifically, we will use
tableaux that are element-wise reversed. The last element in the tableau will
signify the initial state of the TM, the fore-last element will be the
instantaneous description of the TM after a single step, and so on. Reviewing
the proof of Theorem~\ref{T:RPRAM}, it is not difficult to
ascertain that the same results derived for the original tableaux are equally
relevant for order-reversed tableaux.

Under this new definition of tableau, what we are trying to create is a map,
$(\hat{L},\hat{I}, \hat{w} )$, where the topmost $w=s+c-1$ bits of each element
are set to the value $\textit{inp}$, and which meets, in all other respects,
the criteria we previously required of such a map: a number exponential in
$\hat{w}$ of independent, uniformly distributed map elements.

Consider now that the tools which we have already developed
suffice in order to simulate a TM working without input. In this case,
$\textit{inp}=0$, $\textit{inp}\times \hat{I}=0$.

Let us therefore, for now, ignore the original TM that
we intend to simulate and ignore its input. Instead, let us assume that
we have at our disposal a pair $(I_x,w_x)$, where $w_x$ is not equal to
$w_m$, the last $w_i$ value computed by Algorithm~\ref{A:induction_step}, but
rather larger than $(2^{w_m})\leftarrow (2^{w_m})$. (Such a pair
can be found simply by running Algorithm~\ref{A:induction_step} two additional
iterations.) Consider, now, Algorithm~\ref{A:input}, where $w$ is
the width of the tableau vectors to be verified (the bit length of each of their
elements), and $\hat{w}$ is the width of the map that can store it as an
element.

\begin{algorithm}
\caption{Verifying the input}
\label{A:input}
\begin{algorithmic}[1]
\Function{InputVerify}{$I_x,w_x,\hat{w},\textit{inp},w$}
\State $\textit{elementwidth}\Leftarrow 1\leftarrow \hat{w}$\label{Step:ew}
\State $\textit{width}\Leftarrow \textit{elementwidth}\leftarrow \textit{elementwidth}$
\State $L_{\textit{const}}\Leftarrow \textit{RAND}(2^k)\land \textit{MASK}(I_x,\textit{width})$ \label{Step:L_const}
\State $I_{\textit{goodbegin}}\Leftarrow \textit{EQ}(L_{\textit{const}}\land\textit{MASK}(I_x,\textit{elementwidth}),I_x,I_x,\textit{elementwidth})$ \label{Step:const_begin}
\State $I_{\textit{goodtransition}}\Leftarrow (\textit{MASK}(I_{\textit{goodbegin}},\textit{width}) \clr (L_{\textit{const}} \oplus (L_{\textit{const}} \leftarrow \textit{elementwidth}))) \lor \textit{MASK}(I_\textit{goodbegin},\textit{elementwidth})$ \label{Step:const_trans}
\State $I_{\textit{good}}\Leftarrow I_\textit{goodbegin} \land ((I_\textit{goodtransition}+I_\textit{goodbegin}) \rightarrow \textit{width})$ \label{Step:const_good}
\State $L_{\textit{counter}}\Leftarrow \textit{RAND}(2^k)\land\textit{MASK}(I_\textit{good},\textit{width})$ \label{Step:L_counter}
\State $I_{\textit{goodbegin}}\Leftarrow I_\textit{good} \land \textit{EQ}(L_{\textit{counter}}\land\textit{MASK}(I_x,\textit{elementwidth}),0,I_x,\textit{elementwidth})$
\State $\textit{temp} \Leftarrow \textit{ADD}(L_{\textit{counter}}, L_{\textit{const}}, L_{\textit{const}}, \textit{elementwidth}) \leftarrow \textit{elementwidth}$
\State $\textit{gbmask} \Leftarrow \textit{MASK}(I_\textit{goodbegin}, \textit{elementwidth})$
\State $I_{\textit{goodtransition}}\Leftarrow (\textit{MASK}(I_{\textit{goodbegin}}, \textit{width}) \clr (L_{\textit{counter}} \oplus \textit{temp})) \lor \textit{gbmask}$
\State $I_{\textit{good}}\Leftarrow I_\textit{goodbegin} \land ((I_\textit{goodtransition}+I_\textit{goodbegin}) \rightarrow \textit{width})$ \label{Step:counter_good}
\State $M \Leftarrow \textit{MASK}(I_\textit{good} \leftarrow (\textit{inp}\leftarrow (\hat{w}+\hat{w}\minus w)), \textit{elementwidth}\leftarrow (\hat{w}\minus w))$ \label{Step:mask}
\State $L_\textit{output}\Leftarrow L_\textit{counter} \land M$
\State $I_\textit{output}\Leftarrow L_\textit{const} \land M$
\State \Return $(L_\textit{output},I_\textit{output},\hat{w})$\label{Step:ret}
\EndFunction
\end{algorithmic}
\end{algorithm}

Let us analyse this algorithm line by line.

Steps~\ref{Step:L_const} through
\ref{Step:const_good} generate a map
$(L_\textit{const},I_\textit{good},\textit{width})$ that has the bit-string
\[
\left(\mathsf{0}^{\textit{elementwidth}-1}\mathsf{1}\right)^{2^\textit{elementwidth}}
\]
as each one of its elements.
The program performs this by considering all elements in the map
$(L_\textit{const},I_x,\textit{width})$ and filtering out first those indices
whose elements do not begin with the substring
\[
\mathsf{0}^{(\textit{elementwidth}-1)}\mathsf{1}
\]
(Step~\ref{Step:const_begin}), and then those elements which are not composed
entirely of repetitions of a constant string of length $\textit{elementwidth}$.
The latter is tested in Step~\ref{Step:const_trans} by verifying equality
between each substring of length $\textit{elementwidth}$ of $L_\textit{const}$
and the substring of the same length following it immediately. As was done in
Algorithm~\ref{A:induction_step}, an addition operation, carried out on
Step~\ref{Step:const_good}, propagates a carry bit through every element.
The good elements, remaining in the final index set, $I_\textit{good}$, are
those for which the carry propagated through the entire element, thereby
verifying that all the element's bits are correct.

A similar technique, used in Steps~\ref{Step:L_counter} through
\ref{Step:counter_good}, filters $I_\textit{good}$ even further, until it is
known additionally that every element of the map
$(L_\textit{counter},I_\textit{good},\textit{width})$
is a counter of width
$\textit{elementwidth}$. That is to say, its first $\textit{elementwidth}$
bits are all zero, its next $\textit{elementwidth}$ bits are the binary
representation of the number $1$, and so on, in arithmetic progression,
until the last element, being $2^\textit{elementwidth}-1$. This second phase
of filtering on $I_\textit{good}$ is, once again, performed by verifying first
the lowest $\textit{elementwidth}$ bits (which, in this case, must equal $0$)
and then the relation between each element and the next (which is here
incrementation). Simultaneous incrementation of all substrings of length
$\textit{elementwidth}$ is done by adding $L_\textit{const}$ to
$L_\textit{counter}$.

In addition to the two maps generated,
$(L_\textit{const},I_\textit{good},\textit{width})$ and
$(L_\textit{counter},I_\textit{good},\textit{width})$,
consider, now, the following new map:
$(L_\textit{counter},L_\textit{const},\textit{elementwidth})$. If this map
has any elements at all, then it has every possible element of bit-length
$\textit{elementwidth}$. In particular, it would have our desired tableau.

However, while this procedure has so far presented an alternate method for
producing tableau candidates, it still has not addressed the main problem of
verifying that the candidate tableaux begin with the correct bit-string,
$\textit{inp}$.

The method by which this entire construction can now overcome the problem of
verifying $\textit{inp}$ is by noting that in the new structure the value of
each tableau candidate is determined completely by its bit-position relative to
the ``$1$'' bit of $I_\textit{good}$ immediately preceding it. Specifically,
the mask $M$, built in Step~\ref{Step:mask}, is able to mask out all
candidates whose most significant bits do not match the desired value.

Thus, the algorithm generates exactly that subset of the possible tableaux that
have the correct $T[0]$.

Two remarks regarding this algorithm.
\begin{enumerate}
\item In Step~\ref{Step:ew} of the algorithm we define the width of the map
the algorithm constructs. Ostensibly, only a map of width $\hat{w}$ is
required, and there is no need to define $\textit{elementwidth}$ to be any
higher. While this is true for most of Algorithm~\ref{A:input}, we do need
$\textit{elementwidth}$ to be a power of $2$ for step~\ref{Step:mask} to
work properly. The expression
\[
\textit{inp}\leftarrow(\hat{w}+\hat{w}\minus w)
\]
which appears in it is really a rewriting of
\[
(\textit{inp}\times\textit{elementwidth})\leftarrow (\hat{w}\minus w),
\]
which is necessary because ``$\times$'' is not assumed to be available. This
rewrite requires using a known power of $2$ for $\textit{elementwidth}$.
\item The algorithm actually builds the map
$(L_\textit{output},I_\textit{output},\textit{elementwidth})$ with desired
values. However, it returns only
$(L_\textit{output},I_\textit{output},\hat{w})$, truncating the size of its
elements to only $\hat{w}$ bits, as the final output should be. The reason this
can be done is that in all elements, in all bits higher
than $\hat{w}$, the bit values are set to zero by the algorithm. This is the
``desired value'' for these bit positions.
\end{enumerate}

The new algorithm differs from the original one in its error (false rejection)
probabilities. We complete the proof, therefore, by verifying that the new
error probabilities for iteration $i$ of the loop in Step~\ref{Step:infiloop} of
Algorithm~\ref{A:R_sim_TM}, which we denote $p_i$, can still be made to
satisfy $\sum_i p_i \le 1/2$.

Let us bound $p_i$ from above. A false reject occurs in the new algorithm only
in one situation: when $I_\textit{good}=0$ in Step~\ref{Step:counter_good} of
Algorithm~\ref{A:input}. In all other cases, every possible tableau with the
correct initialisation is generated and tested.

The remainder of the argument is the same as in the original proof: for a bit
of $I_\textit{good}$ to be 1, a total of $Z=\BigO(w_x\times\textit{maxstep})$ bits
in a total of $W=\BigO(w_x)$ consecutive bit positions in $\BigO(\textit{maxstep})$
randomly chosen integers are required to attain specific bit values. That
being the case, a choice of $k$ as $W\times(i+1)\times 2^Z$ will ensure
$p_i\approx e^{-(i+1)}$, leading to $\sum_i p_i<1/2$, as desired. Choosing
a larger $k$, so as to avoid the need
for multiplication, only lowers the error probability further.
\end{proof}

\begin{proof}[Proof of Corollary~\ref{C:BPP}]

To extend Corollary~\ref{C:fTIME} from RP acceptance criteria to BPP acceptance
criteria we first note that any RP problem is by definition also a BPP
problem. (Running an RP algorithm twice results in a 0 false acceptance rate
and 1/4 false rejection rate, both of which are better than what is required
by BPP.) We therefore only need to prove that a TM can simulate a RAM under the
appropriate time constraints. Doing so is essentially done as in the proof
of Theorem~\ref{T:RPRAM}. The part of the proof corresponding to
Lemma~\ref{L:BRPRAM} remains unchanged: we use the same simulation of a
randomised RAM by a randomised TM. The only difference is that we use, for both
the RAM and the TM, BPP acceptance criteria, rather than RP ones.

For the equivalent of Lemma~\ref{L:generic_RAND}, showing that a generic random
function is not more powerful than $\textit{RAND}(2^k)$, we use a slightly
different construction.

In the original construction, we picked as $2^{\tilde{k}}$ a value in excess of
$M^{\textit{maxstep}}$. This time we will pick a ${\tilde{k}}$ larger by $2$. Whereas the
original choice of ${\tilde{k}}$ ensured that instead of a false reject probability of
$p_i$ the restricted-$\textit{RAND}$ algorithm will have a false reject of
$p_i/(2-p_i)$, the new choice of ${\tilde{k}}$ now ensures $4p_i/(5-p_i)$.

A little arithmetic now shows that for error probabilities lower than $1/3$
simply running the algorithm $3$ times and taking a majority vote attains
acceptance and rejection error probabilities that are better than the original,
and therefore certainly, over the entire algorithm, within the parameters of
BPP.
\end{proof}

\section{Conclusions}\label{S:conclusions}

This work introduced the new complexity class, PEL, and showed that
PEL arises naturally in both deterministic and randomised PTIME RAM
computations. The power of both the RP-RAM and the BPP-RAM with several
interesting operation sets
was characterised as PEL, this characterisation of the RP-RAM closing a $30$
year old open question.

However, perhaps the most important point of this paper is in pointing out that
$\text{P-RAM}=\text{RP-RAM}$ for some RAMs (specifically, those whose basic
operation sets include division), whereas for others this is not the case.

Although these conclusions seem in no way applicable to the central question of
P vs.\ RP in TMs, it still sheds interesting light on this problem, in pointing
out that the answer of whether stochasticity adds computational power under
RP criteria does not have a single universal answer. Rather, it is
dependent on the details of the computational model examined.

\bibliographystyle{plain}
\bibliography{bibRP_P.bib}

\end{document}